\documentclass[10pt]{article}
\usepackage{fancyhdr}
\usepackage{extramarks}
\usepackage{amsmath}
\usepackage{amsthm}

\usepackage[left]{lineno}

\usepackage[utf8]{inputenc}   
\usepackage[T1]{fontenc}
\usepackage{newunicodechar}   
 \usepackage{dblfloatfix} 
\usepackage{float}

\usepackage{booktabs}
\usepackage{tabularx}

\renewcommand{\arraystretch}{1.2}  

\newunicodechar{α}{$\alpha$}
\newunicodechar{β}{$\beta$}
\newunicodechar{γ}{$\gamma$}
\newunicodechar{κ}{$\kappa$}
\newunicodechar{δ}{$\delta$}

\newunicodechar{–}{--}       
\newunicodechar{—}{---}      
\newunicodechar{’}{'}        
\newunicodechar{±}{$\pm$}

\newunicodechar{­}{}
\usepackage{amsfonts}
\usepackage{xcolor}

\usepackage{siunitx}
\usepackage{tikz}
\usepackage[plain]{algorithm}
\usepackage{algpseudocode}
\usepackage{multirow}
\usepackage{booktabs}
\usepackage{graphicx}
\usepackage{subfigure}
\usepackage[margin=1in]{geometry}
\usepackage{booktabs}
\usepackage{makecell}
\usepackage{array}
\usepackage[colorlinks,linkcolor=black,anchorcolor=black,citecolor=black,urlcolor=blue]{hyperref}
\usepackage{hyphenat}
\usepackage{amsmath,bm}
\usepackage{booktabs}
\usepackage{mathtools}
\usepackage{amssymb}
\usepackage{tikz-cd}
\usepackage{caption}
\usepackage{capt-of}
\usepackage{mciteplus}
\usepackage{cite}
\usepackage{mathrsfs}
\usepackage[title,titletoc,toc]{appendix}
\usepackage{xr}
\usepackage{parskip}
\usepackage{soul}
\usepackage{textcomp}
\usepackage[colaction]{multicol}
\usepackage{lipsum}
\usepackage{etoolbox}
\usepackage{longtable}
\usepackage{array}
\usepackage{tablefootnote}
\usepackage{ragged2e}
\usepackage{soul}
\usepackage{placeins}
\newcolumntype{C}[1]{>{\centering\arraybackslash}p{#1}}
\usetikzlibrary{automata,positioning,arrows.meta,fit,shapes.geometric}

\usepackage{float}

\usetikzlibrary{automata,positioning,arrows.meta,fit,shapes.geometric}

\newtheorem{remark}{Remark}[section]

\newtheorem{prop}{Proposition}[section]

\usepackage{amsthm}  %
\theoremstyle{definition}
\newtheorem{definition}{Definition}[section]

\usepackage{xr}

\begin{document}
	\title{GrassTop: Grassmannian \texorpdfstring{$k$}{k}-mer Topology for Viral Classification and Phylogenetic Analysis}

	\author{Xiang Xiang Wang$^{1}$, Guo-Wei Wei$^{1,2,3}$\thanks{Corresponding author: guowei.wei@uga.edu}
		\\
		$^1$ Department of Mathematics,\\
		University of Georgia, Athens, GA 30602, USA.\\
		$^2$ Department of Biochemistry and Molecular Biology,\\
		University of Georgia, Athens, GA 30602, USA.\\
		$^3$ School of Computing,\\
		University of Georgia, Athens, GA 30602, USA.
	}
	\date{} 

	\maketitle

\begin{abstract}
We introduce GrassTop, a genome representation that integrates Grassmann manifolds and algebraic topology for viral classification and phylogenetic analysis. The framework begins by constructing multiscale topological and spectral descriptors of (k)-mer positional patterns. It then extracts a low-rank subspace that summarizes variation across the filtration and compares genomes using a Grassmannian distance. Although the reported implementation uses the chordal distance, the framework is not restricted to this particular choice.
We evaluate GrassTop on four families of viral classification datasets, four phylogenetic clustering datasets, and a sequence perturbation experiment. Under the reported 5-nearest-neighbor protocol, GrassTop achieves higher scores than five published alignment-free reference methods across all reported classification metrics and datasets. Its UPGMA (unweighted pair-group method using arithmetic averages) trees achieve an average label purity of 1.0 on every phylogenetic dataset. The perturbation experiment provides a more nuanced result: the subspace representation differs most clearly from direct comparison of the unprojected feature matrices for SARS-CoV-2, whereas the differences are smaller or non-monotonic for the other datasets. Overall, these results support GrassTop as an effective topological-geometric representation for viral classification and phylogenetic analysis.
\end{abstract}


\section{Introduction}

Viral classification and phylogenetic analysis are essential for understanding the diversity, evolution, and transmission of viruses. Viral classification provides a systematic framework for organizing viruses based on their genomic, structural, and evolutionary characteristics, enabling the identification of emerging or previously unrecognized viral lineages \cite{murphy2012virus}. Viral phylogenetics reconstructs evolutionary relationships among viral sequences, helping reveal patterns of mutation, recombination, host adaptation, transmission, and geographic spread \cite{gay2024phylogenetic}.  Together, they provide a critical foundation for surveillance, outbreak investigation, prediction of viral evolution, vaccine and antiviral development, and assessment of emerging infectious-disease risks.

Due to the extended lengths of genomes, large language models (LLMs) cannot be directly applied. Effective alignment free methods are commonly used in the field, including Nature Vector Method (NVM) \cite{sun2021geometric}, Feature Frequency Profile with
Jensen-Shannon (FFP-JS) divergence \cite{sims2009alignment,jun2010whole} and with Kullback-Leibler (FFP-KL)  \cite{vinga2003alignment}, Markov $k$-string model (MKS) \cite{qi2004whole}, and Fourier Power Spectrum (FPS) \cite{hoang2015new}.  
CAFE is a model-based alignment-free approach that  further derives explicit evolutionary distances from k-mer frequency statistics under substitution assumptions \cite{lu2017cafe}. 
Although efficient, many existing methods do not adequately capture positional and structural information, which limits their performance in the analysis of genetic variants \cite{zielezinski2017alignment}. 

This work introduces GrassTop as a new geometry and topology-based approach for viral classification and phylogenetic analysis.  GrassTop is a genome representation in which each sequence is mapped to a low-rank subspace on a Grassmannian. The representation is constructed from multiscale algebraic topology  and topological spectral theory of $k$-mer positions along the sequence. It is therefore intended to retain the organization of the descriptor collection across a filtration while replacing alignment comparison by subspace analysis. Once represented, two genomes can be compared with a distance defined on the Grassmannian.

The construction connects two existing ideas. The first is $k$-mer topology, which uses persistent homology \cite{zomorodian2005computing} and persistent Laplacian \cite{wang2020persistent,liu2021persistent} to describe the positional relationship associated genetic sequences \cite{hozumi2024kmertopology}. This topological sequence analysis (TSA) of biological sequences \cite{liu2025topological} is a generalization of the earlier topological data analysis (TDA)   of   biomolecular structures  \cite{xia2014persistent,nguyen2020review, papamarkou2024position}. The $k$-mer topology provides multiple curves over a common filtration and across a selected set of $k$-mer orders, and this construction belongs to a broader line of work applying  algebraic topology, topological spectral theory, and commutative algebra to interpretable analysis of real-world data \cite{ren2025interpretability,wee2025review, cakl2025, su2025topological}. 
Compared with frequency-based $k$-mer methods, $k$-mer topology incorporates crucial positional information.   
The second is subspace-based genome representation via Grassmannians \cite{jost2016geometry}.  Li et al. mapped frequency chaos-game representations to Grassmannians and compared their column spaces geometrically \cite{li2024grassmann}, part of a broader trend toward differential-geometry-based representation learning for biological network and single-cell data \cite{feng2024multiscale}, as well as geometric and geometric-deep-learning models of molecular structure \cite{mu2017geometric,
shen2023moleculargeometric}. GrassTop differs in the
information supplied to the subspace construction: it uses topological and spectral summaries of $k$-mer positional patterns rather than frequency chaos-game representations. Accordingly, the contribution of this paper is not the first use of a Grassmannian for
genomes, but a Grassmannian formulation of multiscale $k$-mer topology descriptors that draws on both of these methodological threads.

GrassTop normalizes and assembles $k$-mer topology descriptor blocks, extracts a rank-$r$
subspace, and represents the genome by that subspace. The experiments instantiate this
general construction with the configuration described in Section~\ref{subsec:datasets_setup}.
We study three questions: whether the
representation supports viral-family classification, whether its distances recover
labelled groups in phylogenetic clustering, and how its response to sequence
perturbations differs from direct comparison of the corresponding descriptors.

The evaluation uses four benchmark NCBI viral-family datasets, four phylogenetic datasets, and
sequence perturbation rates from $10^{-4}$ to $0.4$ \cite{hozumi2024kmertopology}. Classification is assessed with a repeated 5-nearest-neighbor protocol and six metrics; clustering is assessed with UPGMA tree purity \cite{cakl2025}; and the response to perturbation is assessed by label preservation, self preservation, and the normalized distance to the original sequence. These experiments evaluate the complete representation on the stated tasks. They do not by themselves establish that a
Grassmannian representation is preferable for every genome dataset or every choice of input representation.

The remainder of this paper is organized as follows. Section~\ref{sec:foundations} reviews the necessary background on $k$-mer topology features and Grassmannians. Section~\ref{sec:framework} introduces the GrassTop framework that maps a sequence's stacked topological feature blocks onto a Grassmannian point and presents the chordal distance used in the reported experiments. Section~\ref{sec:experiments} describes the experimental setup and reports classification, phylogenetic-clustering, and perturbation results. Section~\ref{sec:discussion} discusses the findings, and Section~\ref{sec:conclusion} concludes the paper.

\section{Mathematical Foundations}
\label{sec:foundations}

This section introduces the mathematical objects used throughout the paper: the $k$-mer
topology feature blocks of Hozumi \& Wei \cite{hozumi2024kmertopology}, and the Grassmannian
on which the GrassTop representation of Section~\ref{sec:framework} is built.

\subsection{\texorpdfstring{$k$}{k}-mer topology features}
\label{subsec:kmer-topology}

The input features of GrassTop are derived from the $k$-mer topology construction of
Hozumi and Wei \cite{hozumi2024kmertopology}. The following formulation summarizes the
components of that construction required for the Grassmannian representation.

Let $S=s_1s_2\cdots s_N$, with $s_i\in\{A,C,G,T\}$, and let $\ell$ be one of the
$4^k$ possible $k$-mer types. If $\ell$ occurs $n_\ell$ times, Hozumi and Wei define
the $k$-mer-specific position function
\cite[Definition~4.2]{hozumi2024kmertopology}
\[
 S^\ell:[n_\ell]\longrightarrow\mathbb{N},\qquad i\longmapsto S^\ell(i),
\]
where $S^\ell(i)$ is the starting position of the $i$th occurrence of $\ell$ and
$[n_\ell]=\{1,\ldots,n_\ell\}$. The associated $k$-mer-specific position-distance
matrix is
\begin{equation}
 D^\ell=\{d_{ij}^\ell\},\qquad
 d_{ij}^\ell=\left|S^\ell(i)-S^\ell(j)\right|,
 \quad 1\leq i,j\leq n_\ell,
 \label{eq:kmer-position-distance}
\end{equation}
as in \cite[Section~4.1.1]{hozumi2024kmertopology}. Thus, the input records pairwise
separations between occurrences of the same $k$-mer rather than only its total count.

\subsubsection{Persistent Laplacian features}
The filtration radius is denoted by $\rho$ rather than by $r$, as in
\cite{hozumi2024kmertopology}, because $r$ denotes the retained Grassmannian rank below.
For $\rho\geq0$, Hozumi and Wei construct the
$\ell$-specific degree-zero graph Laplacian
\cite[Section~4.1.2]{hozumi2024kmertopology}
\begin{equation}
 (L_\rho^\ell)_{ij}=
 \begin{cases}
 -1, & i\neq j\ \text{and}\ d_{ij}^\ell\leq\rho,\\
 \#\{j\neq i:d_{ij}^\ell\leq\rho\}, & i=j,\\
 0, & \text{otherwise}.
 \end{cases}
 \label{eq:kmer-graph-laplacian}
\end{equation}
Equivalently, $L_\rho^\ell=\mathsf{D}_\rho^\ell-A_\rho^\ell$, the standard
combinatorial Laplacian of the threshold graph at radius $\rho$
\cite{chung1997spectral}. These graphs are nested as $\rho$ increases, giving the
filtration used for degree-zero persistent homology
\cite{edelsbrunner2002topological,zomorodian2005computing}. The eigenvalues are ordered as
\[
 0=\lambda_1^\ell(\rho)\leq\cdots\leq
 \lambda_{n_\ell}^\ell(\rho).
\]
The multiplicity of zero, rather than the number of nonzero eigenvalues, is the
Betti-0 number:
\begin{equation}
 \beta_0^\ell(\rho)=\dim\ker L_\rho^\ell.
 \label{eq:kmer-beta-zero}
\end{equation}
This equality follows from the standard graph-Laplacian kernel theorem
\cite{chung1997spectral,wang2020persistent}. Let
$m_\ell(\rho)=n_\ell-\beta_0^\ell(\rho)$ and write the positive eigenvalues as
$0<\lambda_1^{\ell,+}(\rho)\leq\cdots\leq
\lambda_{m_\ell(\rho)}^{\ell,+}(\rho)$. Two descriptors considered by Hozumi and
Wei are
\[
 B^\ell(\rho)=\beta_0^\ell(\rho),\qquad
 E^\ell(\rho)=\lambda_1^{\ell,+}(\rho),
\]
collected over increasing radii to form $k$-mer-specific curves
\cite[Section~4.1.2]{hozumi2024kmertopology}.

The present study additionally considers the mean and standard deviation of the positive
spectrum. These quantities are not definitions introduced by Hozumi and Wei; they are
persistent spectral summary statistics of the type described by Meng and Xia
\cite{meng2021perspect}:
\begin{align}
 \mathrm{MEANP}^\ell(\rho)
 &=\frac{1}{m_\ell(\rho)}
   \sum_{j=1}^{m_\ell(\rho)}\lambda_j^{\ell,+}(\rho),\\
 \mathrm{STDP}^\ell(\rho)
 &=\left[
   \frac{1}{m_\ell(\rho)}
   \sum_{j=1}^{m_\ell(\rho)}
   \bigl(\lambda_j^{\ell,+}(\rho)-\mathrm{MEANP}^\ell(\rho)\bigr)^2
   \right]^{1/2}.
 \label{eq:kmer-spectral-moments}
\end{align}
The three positive-spectrum descriptors are set to zero when no positive eigenvalue
exists; all four descriptors are set to zero when $\ell$ is absent.

\begin{definition}[$k$-mer topology feature block]
\label{def:kmer-block}
For each descriptor $g$ selected for analysis and sampled at $T$ increasing radii
$\rho_1^{(k)},\ldots,\rho_T^{(k)}$, the $4^k$ $k$-mer-specific curves are arranged into
\[
G_{S,g}^{(k)} \in \mathbb{R}^{4^k \times T},
\qquad
\bigl[G_{S,g}^{(k)}\bigr]_{\ell,t} = g^\ell\bigl(\rho_t^{(k)}\bigr),
\]
whose rows follow the lexicographic order of the $k$-mer types and whose columns follow
the filtration order.
\end{definition}

Definition~\ref{def:kmer-block} introduces only the matrix arrangement used by
GrassTop; the position distances, Laplacians, and source curves are those defined above.
The construction is not restricted to a prescribed set of spectral descriptors. In the
experiments reported here, the selected descriptors are
$g\in\{B,E,\mathrm{MEANP},\mathrm{STDP}\}$.

\subsection{Grassmannians and subspace distances}
\label{sec:methods-grassmann}

For integers $1\leq r<n$, the collection of $r$-dimensional linear subspaces of
$\mathbb{R}^n$ forms a Grassmannian. Its representation by orthonormal frames,
the identification of subspaces through the singular value decomposition, and the
comparison of subspaces through principal angles are summarized below
\cite{edelman1998geometry,ye2016schubert}.

\begin{definition}[Grassmannian \cite{edelman1998geometry,ye2016schubert}]
\label{def:grassmann}
The Grassmannian $\mathrm{Gr}(r,n)$ is the set of all $r$-dimensional linear
subspaces of $\mathbb{R}^n$. If
\[
 \mathrm{St}(r,n)=\{U\in\mathbb{R}^{n\times r}:U^\top U=I_r\}
\]
is the Stiefel manifold of orthonormal $r$-frames, then
\[
 \mathrm{Gr}(r,n)\cong\mathrm{St}(r,n)/\mathrm{O}(r).
\]
Consequently, $U,U'\in\mathrm{St}(r,n)$ represent the same Grassmannian point if and
only if $U'=UO$ for some $O\in\mathrm{O}(r)$.
\end{definition}

\subsubsection{Singular value decomposition}
Let $Z\in\mathbb{R}^{n\times T}$ and $1\leq r<\min(n,T)$. A singular value
decomposition of $Z$ is \cite{golub2013matrix}
\[
 Z=U\Sigma V^\top,\qquad
 \sigma_1(Z)\geq\cdots\geq\sigma_{\min(n,T)}(Z)\geq0,
\]
where $U$ and $V$ are orthogonal matrices and the singular values are ordered
nonincreasingly. The first $r$ left singular vectors form
$U_r\in\mathrm{St}(r,n)$, and their span defines the left singular subspace associated
with the $r$ largest singular values:
\begin{equation}
 \mathcal{U}_r(Z)=\operatorname{span}(U_r)\in\mathrm{Gr}(r,n).
 \label{eq:grassmann-map}
\end{equation}
The truncated factorization $U_r\Sigma_rV_r^\top$ is a best rank-$r$ approximation of
$Z$ in the Frobenius norm by the Eckart--Young theorem
\cite{eckart1936approximation,golub2013matrix}. The subspace $\mathcal{U}_r(Z)$ is
uniquely determined when $\sigma_r(Z)>\sigma_{r+1}(Z)$; if the two singular values are
equal, the choice of an $r$-dimensional subspace within the corresponding singular
subspace need not be unique \cite{golub2013matrix}.

\subsubsection{Principal angles and chordal distance}
Distances between Grassmannian points can be expressed through principal angles. For
$\mathcal{U}=\operatorname{span}(U)$ and $\mathcal{V}=\operatorname{span}(V)$ in
$\mathrm{Gr}(r,n)$, their principal angles
$0\leq\theta_1\leq\cdots\leq\theta_r\leq\pi/2$ satisfy
\begin{equation}
 \cos\theta_i=\sigma_i(U^\top V),\qquad i=1,\ldots,r,
 \label{eq:principal-angles}
\end{equation}
where the singular values are ordered nonincreasingly
\cite{bjorck1973angles,ye2016schubert}. Principal angles depend only on the two
subspaces, not on the particular orthonormal bases used to represent them.

The chordal distance between two points $\mathcal{U},\mathcal{V}\in\mathrm{Gr}(r,n)$
is defined in terms of their principal angles \cite[Table~2]{ye2016schubert}. For
orthonormal basis matrices $U$ and $V$, its equivalent principal-angle and
projection-matrix forms are
\begin{equation}
  d_{\mathrm{chord}}(\mathcal{U},\mathcal{V})
  =\left(\sum_{i=1}^{r}\sin^2\theta_i\right)^{1/2}
  =\frac{1}{\sqrt{2}}
   \left\lVert UU^\top-VV^\top\right\rVert_F,
  \label{eq:chordal}
\end{equation}
where $\lVert A\rVert_F=(\operatorname{tr}(A^\top A))^{1/2}$ is the Frobenius norm.
Indeed,
\begin{equation}
 \left\lVert UU^\top-VV^\top\right\rVert_F^2
 =2r-2\left\lVert U^\top V\right\rVert_F^2
 =2\sum_{i=1}^{r}\sin^2\theta_i.
 \label{eq:chordal-equivalence}
\end{equation}
It follows that
$0\leq d_{\mathrm{chord}}(\mathcal{U},\mathcal{V})\leq\sqrt{r}$, with distance zero
if and only if the two subspaces coincide. Replacing the bases $U$ and $V$ by $UO_1$
and $VO_2$, respectively, where $O_1,O_2\in\mathrm{O}(r)$, leaves the corresponding
projection matrices and the distance unchanged. Thus, the chordal distance is
independent of the orthonormal bases chosen to represent the two Grassmannian points
\cite{ye2016schubert}.
For reference, Table~\ref{tab:notation} summarizes the main notation used throughout the
remainder of the paper.

\begin{table}[H]
\centering
\scriptsize
\setlength{\tabcolsep}{4pt}
\renewcommand{\arraystretch}{1.0}
\caption{Summary of main notation used throughout the paper.}
\label{tab:notation}
\begin{tabular}{cl}
\toprule
\textbf{Symbol} & \textbf{Description} \\
\midrule
$S$ & A genomic sequence \\
$\mathcal{K}\subset\mathbb{N}_{>0}$ & Finite, nonempty set of $k$-mer orders \\
$k\in\mathcal{K}$ & A $k$-mer order \\
$m_{\mathcal{K}}=\sum_{k\in\mathcal{K}}4^k$ & Number of $k$-mer words across these orders \\
$T$ & Number of filtration steps \\
$\mathcal{G}$ & Finite, nonempty set of descriptor types \\
$g\in\mathcal{G}$ & A descriptor type (Def.~\ref{def:kmer-block}) \\
$G_{S,g}^{(k)}\in\mathbb{R}^{4^k\times T}$ & $k$-mer topology feature block (Def.~\ref{def:kmer-block}) \\
$s_{g,k}$ & Dataset-level RMS normalization scale \\
$c_{g,k}=\sqrt{w_k}/\max(s_{g,k},\varepsilon)$ & Normalization weight for descriptor $g$ and order $k$ \\
$q=|\mathcal{G}|$ & Number of descriptor types \\
$n=q m_{\mathcal{K}}$ & Ambient feature dimension \\
$Z_S\in\mathbb{R}^{n\times T}$ & Normalized, stacked feature matrix \\
$r$ & Retained rank \\
$Z_S=U_S\Sigma_SV_S^\top$ & Truncated singular value decomposition of $Z_S$ \\
$\mathrm{Gr}(r,n)$ & Grassmannian of $r$-dim.\ subspaces of $\mathbb{R}^n$ (Def.~\ref{def:grassmann}) \\
$\mathrm{span}(U_S)\in\mathrm{Gr}(r,n)$ & GrassTop representation of $S$ \\
$d_{\mathrm{chord}}(S,R)$ & Grassmannian chordal distance (Eq.~\ref{eq:chordal}) \\
$\lVert Z_S-Z_R\rVert_F$ & Frobenius distance between unprojected feature matrices \\
\bottomrule
\end{tabular}
\end{table}

\FloatBarrier
\section{The GrassTop Framework}
\label{sec:framework}

GrassTop is a framework for representing genomes as subspaces derived from multiscale
topological and spectral descriptors of $k$-mer positional patterns. For each genome,
the descriptor blocks are normalized and assembled before a retained singular subspace
is extracted. This subspace defines a point on a Grassmannian, so two genomes can be
compared through a distance between their corresponding Grassmannian points.
Figure~\ref{fig:framework} presents the overall framework. Section~\ref{subsec:grasstop-construction}
defines the construction, Section~\ref{subsec:geometric-properties} establishes its main
geometric properties, and Section~\ref{subsec:computational_cost} describes its
computational cost.

\subsection{GrassTop construction}
\label{subsec:grasstop-construction}

Let $\mathcal{K}\subset\mathbb{N}_{>0}$ denote a finite, nonempty set of $k$-mer orders,
and let $\mathcal{G}$ denote a finite, nonempty set of descriptor types as defined in
Section~\ref{subsec:kmer-topology}.
Starting from the $k$-mer topology descriptors of
Section~\ref{subsec:kmer-topology}, GrassTop assembles the descriptor blocks, extracts a
rank-$r$ subspace, and compares the resulting genome representations with a distance
$d_{\mathrm{Gr}}$ on the Grassmannian. The construction has four steps: feature
extraction, normalization and stacking, Grassmannian embedding, and subspace comparison.

\begin{figure}[ht]
\centering
\includegraphics[width=\textwidth]{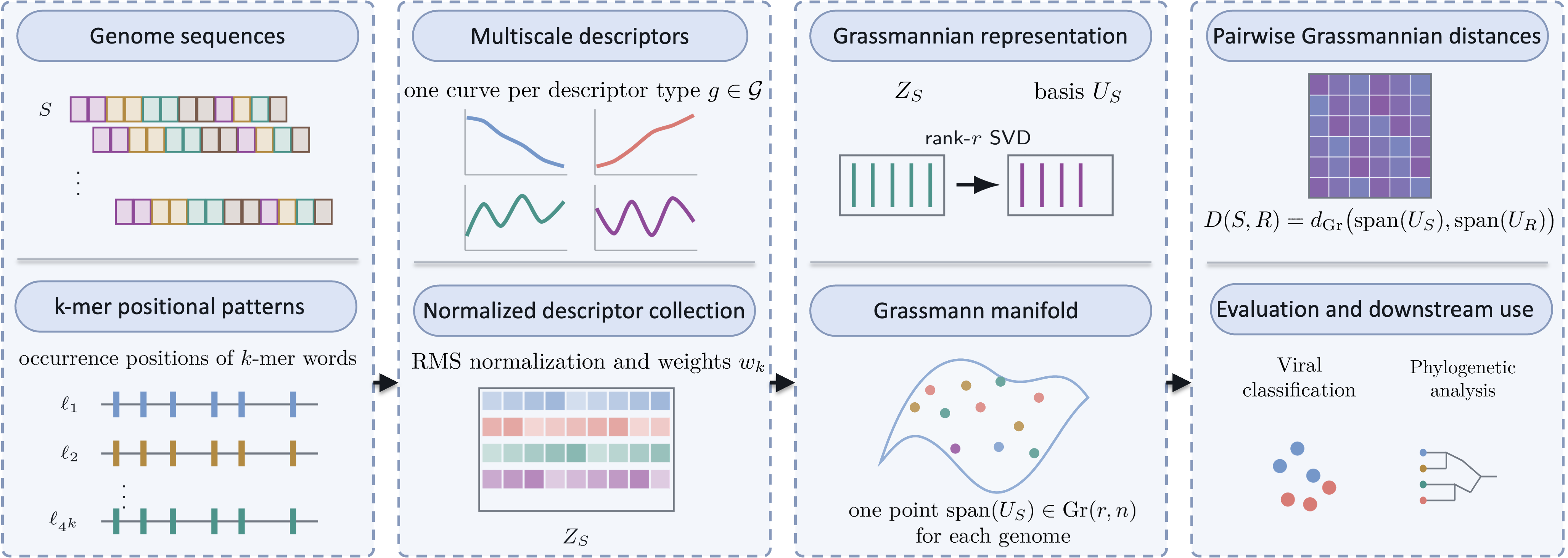}
\caption{Overall GrassTop framework. For each genome, the occurrence positions of
$k$-mer words are used to compute multiscale descriptors. The normalized descriptor
collection is represented by a rank-$r$ subspace on a Grassmann manifold. Pairwise
Grassmannian distances between these subspaces provide the input for classification and
phylogenetic analysis.}
\label{fig:framework}
\end{figure}
\FloatBarrier

\medskip
\noindent\textbf{Step 1 (Feature extraction).}
For a sequence $S$, the $k$-mer topology feature blocks $G_{S,g}^{(k)}$ are computed for
$k\in\mathcal{K}$ and $g\in\mathcal{G}$, as defined in
Section~\ref{subsec:kmer-topology}. Each block records one descriptor type over the
common filtration steps.

\medskip
\noindent\textbf{Step 2 (Normalization and stacking).}
Each block $G_{S,g}^{(k)}$ is normalized by an RMS scale $s_{g,k}$ and a
$k$-dependent weight $c_{g,k}=\sqrt{w_k}/\max(s_{g,k},\varepsilon)$, where
$w_k>0$ is the block weight for order $k$: it sets how much each $k$-mer
order contributes to the assembled matrix $Z_S$ relative to the other orders,
independently of the RMS rescaling. The values of $w_k$ used in the
experiments are given in Section~\ref{subsec:datasets_setup}. The normalized
blocks are then vertically stacked over
$(g,k)$ into one matrix
\[
Z_S \in \mathbb{R}^{q m_{\mathcal{K}} \times T},
\]
where $q=|\mathcal{G}|$ and $m_{\mathcal{K}}=\sum_{k\in\mathcal{K}}4^k$. Let
$n=q m_{\mathcal{K}}$ denote this ambient feature dimension. The scaling balances blocks from different descriptor types and $k$-mer
orders before the subspace is extracted.

\medskip
\noindent\textbf{Step 3 (Rank-$r$ Grassmannian embedding).}
For a retained rank $1\le r\le\min(n,T)$, let
$Z_S = U_S \Sigma_S V_S^\top$ denote a rank-$r$ truncated singular value decomposition.
GrassTop maps $S$ to
$\mathrm{span}(U_S) \in \mathrm{Gr}(r,n)$, the subspace spanned by the leading
left singular vectors. The singular values $\Sigma_S$ and right singular vectors $V_S$
are not part of the Grassmannian representation.

\medskip
\noindent\textbf{Step 4 (Grassmannian distance).}
Let $d_{\mathrm{Gr}}$ be a distance on $\mathrm{Gr}(r,n)$. The distance
between sequences $S$ and $R$ is defined as
\[
  D(S,R)=d_{\mathrm{Gr}}\bigl(\mathrm{span}(U_S),\mathrm{span}(U_R)\bigr).
\]
For the chordal distance of Eq.~\eqref{eq:chordal}, this becomes
$D(S,R)=d_{\mathrm{chord}}(S,R)$. Other Grassmannian distances can be used without
changing the representation.

\begin{remark}[Framework specification]
An instance of GrassTop is specified by its $k$-mer orders, descriptor types,
filtration, block weights, retained rank, normalization scales, and Grassmannian
distance. These quantities are fixed before genomes are represented or compared. Under
a fixed specification, the mapping and pairwise comparison are deterministic and
contain no learned weights. The specification used in the experiments is given in
Section~\ref{subsec:datasets_setup}.
\end{remark}

\subsection{Geometric properties of the representation}
\label{subsec:geometric-properties}

The GrassTop representation is determined by the leading left singular subspace of $Z$
rather than by a particular singular-vector basis. This construction yields the
invariance and equivariance properties stated below. Define
\[
  \Phi_r(Z)=\mathrm{span}(U_r),
\]
where $U_r$ contains the $r$ leading left singular vectors of $Z$.

\begin{prop}[Invariance and equivariance of the GrassTop embedding]
\label{prop:grasstop-invariance}
Let $Z\in\mathbb{R}^{n\times T}$ and $1\le r<\min(n,T)$ satisfy
$\sigma_r(Z)>\sigma_{r+1}(Z)$, so that its leading $r$-dimensional left singular
subspace is unique. For any nonzero scalar $a$, any $Q\in\mathrm{O}(T)$, and any
$A\in\mathrm{O}(n)$,
\begin{align}
  \Phi_r(aZQ) &= \Phi_r(Z), \label{eq:right-invariance}\\
  \Phi_r(AZ) &= A\Phi_r(Z). \label{eq:left-equivariance}
\end{align}
Consequently, independent nonzero rescaling and right-orthogonal transformation of two
descriptor matrices do not change their GrassTop chordal distance. A common left-orthogonal
transformation also leaves that distance unchanged.
\end{prop}

\begin{proof}
The left singular subspace of $Z$ is the eigenspace of $ZZ^\top$ associated with its $r$
largest eigenvalues. Since
\[
  (aZQ)(aZQ)^\top=a^2ZZ^\top,
\]
the leading eigenspace is unchanged, which proves Eq.~\eqref{eq:right-invariance}. Also,
$(AZ)(AZ)^\top=A(ZZ^\top)A^\top$, so its leading eigenspace is the image under $A$ of the
leading eigenspace of $ZZ^\top$, proving Eq.~\eqref{eq:left-equivariance}. The distance
statements follow from Eq.~\eqref{eq:chordal} and the invariance of the Frobenius norm
under orthogonal transformations.
\end{proof}

The proposition also clarifies the role of normalization. GrassTop is invariant to
overall scaling of a descriptor matrix, but not to
arbitrary rescaling of individual rows or descriptor blocks; such transformations can change
$ZZ^\top$ and hence the represented subspace. Block normalization and the retained
rank are therefore structural parts of the representation rather than properties removed
by the Grassmannian embedding.

\subsection{Computational cost}
\label{subsec:computational_cost}

GrassTop uses the persistent-homology and persistent-Laplacian calculations of
\cite{hozumi2024kmertopology} to obtain the $k$-mer topology descriptors. Their cost
depends on the $k$-mer orders, word occurrences and positions, filtration resolution,
and the topological or spectral quantities being computed. The analysis below focuses
on the normalization, assembly, subspace computation, and distance evaluation applied
after descriptor construction.

Let $n=q m_{\mathcal{K}}$ be the number of rows in the assembled descriptor array and let
$T$ be the number of filtration values. For each genome, normalization and assembly require
$O(nT)$ operations. The implementation used in this study forms the $T\times T$ Gram
matrix $Z_S^\top Z_S$, computes its eigendecomposition, and recovers the $r$ leading left
singular vectors. These operations require, respectively, $O(nT^2)$, $O(T^3)$, and
$O(nTr)$ time. The working memory for the assembled array, Gram matrix, and retained basis
is $O(nT+T^2+nr)$. In the rank-deficient case, the implementation instead uses a dense
SVD, whose time complexity is $O(nT\min(n,T))$.

Once the subspace representations have been constructed, computing the chordal distance
between two genomes requires $O(nr^2)$ operations. For a collection of $N$ genomes, all
pairwise distances therefore require $O(N^2nr^2)$ operations and $O(N^2)$ storage if the
complete distance matrix is retained. Estimating shared normalization factors requires a
single pass over the relevant descriptor arrays. After these factors have been fixed,
feature construction and subspace computation can be performed independently for each
genome.

\FloatBarrier
\section{Experiments}
\label{sec:experiments}

We evaluated GrassTop through viral-family classification, phylogenetic clustering, and
sequence perturbation. For classification, pairwise chordal distances were used for
5-nearest-neighbor prediction of viral-family labels. For phylogenetic analysis, the same
distances were used to construct UPGMA trees, which were evaluated against known taxonomic
or lineage groups. The perturbation experiment changed bases at specified rates
and compared how well GrassTop and the underlying descriptors preserved source identity
and group labels. The classification and clustering results were compared with published
alignment-free methods.

\subsection{Datasets and evaluation}
\label{subsec:datasets_setup}

All experiments use the same $k$-mer orders, filtration resolution, block weights,
retained rank, and Grassmannian distance. Specifically, $\mathcal{K}=\{1,\ldots,5\}$,
$T=50$ filtration steps,
$w_k=2^{-(5-k)}$ \cite{hozumi2024kmertopology}, rank $r=10$, and the chordal distance.
For this configuration, $m_{\mathcal{K}}=\sum_{k=1}^{5}4^k=1364$. Classification uses the four
descriptor types $B$, $E$, $\mathrm{MEANP}$, and $\mathrm{STDP}$; phylogenetic clustering uses
$B$ alone; and the perturbation experiment uses $B$, $\mathrm{MEANP}$, and $\mathrm{STDP}$.
For each task, the RMS scales $s_{g,k}$ are estimated once from the complete dataset
without class labels and are then held fixed throughout the analysis, including all
cross-validation folds in the classification task.

\subsubsection{Classification datasets and evaluation}
\label{sec:eval-classification}

Four NCBI viral-genome datasets were used: NCBI 2020, NCBI 2022, NCBI 2024,
and NCBI 2024 (All) \cite{hozumi2024kmertopology}. Families containing fewer than
15 sequences were removed before the data were split. Table~\ref{tab:datasets} gives
the numbers of sequences and families remaining in each dataset.

\begin{table}[ht]
\centering
\caption{Dataset sizes after removing families with fewer than 15 sequences.}
\label{tab:datasets}
\begin{tabular}{lrr}
\toprule
Dataset & Sequences & Families \\
\midrule
NCBI 2020        & 6{,}810  & 57  \\
NCBI 2022        & 11{,}066 & 78  \\
NCBI 2024        & 11{,}635 & 109 \\
NCBI 2024 (All)  & 13{,}160 & 120 \\
\bottomrule
\end{tabular}
\end{table}

Each filtered dataset was evaluated by stratified 5-fold cross-validation with 30 random
seeds. For each test sequence, the five training sequences with the smallest GrassTop
distances were identified, and their majority family was assigned to the test sequence.
Accuracy (ACC), balanced accuracy (BA), F1, one-vs-rest AUC-ROC, recall, and precision
were calculated for each fold. Results were averaged over the 150 cross-validation folds.
The reference methods were the natural vector method (NVM) \cite{sun2021geometric},
feature-frequency profiles compared by Jensen--Shannon divergence (FFP-JS)
\cite{sims2009alignment,jun2010whole}, feature-frequency profiles compared by
Kullback--Leibler divergence (FFP-KL) \cite{vinga2003alignment}, the Markov $k$-string
method \cite{qi2004whole}, and the Fourier power spectrum method (FPS)
\cite{hoang2015new}. Their values in Section~\ref{sec:classification} were taken from
\cite{hozumi2024kmertopology}, which used the same evaluation procedure.

\subsubsection{Phylogenetic datasets and tree evaluation}
\label{sec:eval-phylogenetics}

Four datasets were used for phylogenetic clustering: human rhinovirus
(HRV, $n{=}116$, four labels), mammalian
mitochondrial genomes ($n{=}41$, eight labels), Ebolavirus ($n{=}59$, five labels), and
SARS-CoV-2 ($n{=}44$, eight lineage groups) \cite{cakl2025}. The HRV, mitochondrial, and
Ebolavirus datasets contain every accession in the published lists. The 44 SARS-CoV-2
sequences and their lineage labels were retrieved
from GISAID using the published EPI\_ISL accessions.

For each dataset, GrassTop distances were used to construct a UPGMA tree. Labels were
not used to calculate distances or construct the tree. Tree purity was calculated from
Eqs.~22--24 of Ref. \cite{cakl2025}. Let $\mathcal{L}$ be
the set of labels, let $S_\ell$ be the set of leaves with label $\ell$, and let
$n_\ell=|S_\ell|$. The maximal pure subtrees for label $\ell$ partition $S_\ell$ into
$S_{\ell,1},\ldots,S_{\ell,m_\ell}$. The tree-level purity reported in Table~\ref{tab:purity}
is
\begin{equation}
 \operatorname{Purity}(T)
 =\frac{1}{|\mathcal{L}|}
 \sum_{\ell\in\mathcal{L}}
 \sum_{j=1}^{m_\ell}
 \left(\frac{|S_{\ell,j}|}{n_\ell}\right)^2.
 \label{eq:tree-purity}
\end{equation}
A label has purity 1 when all of its leaves belong to one pure subtree; fragmentation
across several pure subtrees gives a smaller value.

The five Ebolavirus reference entries marked $\dagger$ in Table~\ref{tab:purity} use the
clade-consistency precision reported by \cite{hozumi2024kmertopology}, rather than
Eq.~\eqref{eq:tree-purity}. For a label $\ell$, let $C_\ell$ be the leaf set of the
smallest clade containing all leaves in $S_\ell$. Its precision is
\begin{equation}
 \operatorname{Prec}(\ell)=\frac{|S_\ell|}{|C_\ell|}.
 \label{eq:clade-precision}
\end{equation}
Thus, the score is 1 for a label exactly when its smallest containing clade has no leaves
from other labels. The cited study reports a value of 1.00 for each of the five reference
methods on Ebolavirus. These values are retained as published but are not directly
comparable with the GrassTop purity because the scoring definitions differ.

\subsubsection{Sequence perturbation experiment}
\label{sec:eval-stability}

The same four phylogenetic datasets were used in the sequence perturbation experiment.
At each perturbation rate $p$, every A, C, G, or T position was independently
changed with probability $p$. The original base at each selected position was
replaced with one of the other three bases, chosen uniformly at random. The rates were
\[
p \in \{0.0001, 0.0005, 0.001, 0.005, 0.01, 0.02, 0.05, 0.1, 0.15, 0.25, 0.4\}.
\]
For example, at $p=0.001$, a sequence of length 10,000 has 10 changed positions on
average. If position 237 is selected and contains A, it may be changed to G, while every
position not selected remains unchanged.

One original sequence was sampled from each label. At each rate, the procedure above was
repeated independently five times for every sampled sequence. The selected positions and
replacement bases could therefore differ among the five perturbed sequences. No bases
were inserted or deleted, so all sequences retained their original length. For each
dataset and rate, this gave 20 perturbed
sequences for HRV, 25 for Ebolavirus, and 40 each for the mitochondrial and SARS-CoV-2
datasets. GrassTop and the
descriptors before subspace projection were computed from the same perturbed sequences
using $B$, $\mathrm{MEANP}$, and $\mathrm{STDP}$. For both GrassTop and these descriptors,
the five nearest neighbors of each perturbed sequence were selected from all original
sequences.

\medskip
\noindent\textbf{Label-preservation.}
A perturbed sequence is counted as label-preserved when the majority label among its five
nearest neighbors equals the label of its original sequence. The reported score is the
fraction of perturbed sequences satisfying this condition at a given perturbation rate.

\medskip
\noindent\textbf{Self-preservation.}
A perturbed sequence is counted as self-preserved when its original sequence occurs among
its five nearest neighbors. The reported score is the fraction of perturbed sequences
satisfying this condition at a given perturbation rate.

\medskip
We also measured the distance from each perturbed sequence to its original sequence. We
used the chordal distance for GrassTop (Eq.~\eqref{eq:chordal}) and the Frobenius distance for the
descriptors before projection. Each distance was divided by the mean pairwise distance
among the original sequences, calculated separately for GrassTop and
the descriptors. A normalized value of 1 therefore equals the mean distance between
original sequences for the corresponding method. The same original and perturbed
sequences were used for all three measurements.

\subsection{Results}
\label{sec:results}

The following subsections present the experimental results for viral-family
classification, phylogenetic clustering, and sequence perturbation.

\subsubsection{Classification benchmark}
\label{sec:classification}

GrassTop has the largest displayed value for each of the six classification metrics in
all four datasets (Tables~\ref{tab:ncbi2020}--\ref{tab:ncbi2024full}). The reference values
for NVM \cite{sun2021geometric}, FFP-JS \cite{sims2009alignment,jun2010whole}, FFP-KL
\cite{vinga2003alignment}, Markov $k$-string \cite{qi2004whole}, and FPS
\cite{hoang2015new} were taken from Table~1 of \cite{hozumi2024kmertopology}.

\begin{table}[ht]
\centering
\caption{5-NN classification results (macro-averaged over 5-fold $\times$ 30-seed CV),
NCBI 2020. The first five rows are reported in Table~1 of
\cite{hozumi2024kmertopology}. Bold indicates the largest displayed value in each column.}
\label{tab:ncbi2020}
\setlength{\tabcolsep}{4pt}
\begin{tabular}{lrrrrrr}
\toprule
Method & ACC & BA & F1 & AUC-ROC & Recall & Precision \\
\midrule
NVM & 0.847 & 0.807 & 0.809 & 0.960 & 0.807 & 0.840 \\
FFP-JS & 0.821 & 0.790 & 0.781 & 0.954 & 0.790 & 0.814 \\
FFP-KL & 0.819 & 0.789 & 0.780 & 0.954 & 0.789 & 0.814 \\
Markov & 0.713 & 0.644 & 0.622 & 0.905 & 0.644 & 0.668 \\
FPS & 0.714 & 0.637 & 0.633 & 0.917 & 0.637 & 0.665 \\
\midrule
GrassTop & \textbf{0.9131} & \textbf{0.8909} & \textbf{0.8925} & \textbf{0.9798} & \textbf{0.8909} & \textbf{0.9046} \\
\bottomrule
\end{tabular}
\end{table}

\begin{table}[H]
\centering
\caption{5-NN classification results, NCBI 2022. The first five rows are reported in
Table~1 of \cite{hozumi2024kmertopology}. Bold indicates the largest displayed value in
each column.}
\label{tab:ncbi2022}
\setlength{\tabcolsep}{4pt}
\begin{tabular}{lrrrrrr}
\toprule
Method & ACC & BA & F1 & AUC-ROC & Recall & Precision \\
\midrule
NVM & 0.852 & 0.747 & 0.750 & 0.935 & 0.747 & 0.791 \\
FFP-JS & 0.830 & 0.740 & 0.733 & 0.937 & 0.740 & 0.769 \\
FFP-KL & 0.832 & 0.743 & 0.735 & 0.934 & 0.743 & 0.771 \\
Markov & 0.724 & 0.593 & 0.577 & 0.876 & 0.593 & 0.617 \\
FPS & 0.723 & 0.588 & 0.580 & 0.890 & 0.588 & 0.603 \\
\midrule
GrassTop & \textbf{0.9040} & \textbf{0.8229} & \textbf{0.8288} & \textbf{0.9593} & \textbf{0.8229} & \textbf{0.8656} \\
\bottomrule
\end{tabular}
\end{table}

\begin{table}[H]
\centering
\caption{5-NN classification results, NCBI 2024. The first five rows are reported in
Table~1 of \cite{hozumi2024kmertopology}. Bold indicates the largest displayed value in
each column.}
\label{tab:ncbi2024}
\setlength{\tabcolsep}{4pt}
\begin{tabular}{lrrrrrr}
\toprule
Method & ACC & BA & F1 & AUC-ROC & Recall & Precision \\
\midrule
NVM & 0.814 & 0.729 & 0.738 & 0.933 & 0.729 & 0.789 \\
FFP-JS & 0.796 & 0.724 & 0.712 & 0.936 & 0.724 & 0.744 \\
FFP-KL & 0.796 & 0.727 & 0.714 & 0.932 & 0.727 & 0.747 \\
Markov & 0.633 & 0.589 & 0.554 & 0.876 & 0.589 & 0.573 \\
FPS & 0.660 & 0.561 & 0.560 & 0.882 & 0.561 & 0.593 \\
\midrule
GrassTop & \textbf{0.8805} & \textbf{0.8135} & \textbf{0.8228} & \textbf{0.9590} & \textbf{0.8135} & \textbf{0.8645} \\
\bottomrule
\end{tabular}
\end{table}

\begin{table}[H]
\centering
\caption{5-NN classification results, NCBI 2024 (All). The first five rows are reported
in Table~1 of \cite{hozumi2024kmertopology}. Bold indicates the largest displayed value
in each column.}
\label{tab:ncbi2024full}
\setlength{\tabcolsep}{4pt}
\begin{tabular}{lrrrrrr}
\toprule
Method & ACC & BA & F1 & AUC-ROC & Recall & Precision \\
\midrule
NVM & 0.809 & 0.729 & 0.738 & 0.931 & 0.729 & 0.788 \\
FFP-JS & 0.799 & 0.721 & 0.711 & 0.935 & 0.721 & 0.745 \\
FFP-KL & 0.799 & 0.724 & 0.712 & 0.931 & 0.724 & 0.745 \\
Markov & 0.638 & 0.589 & 0.555 & 0.876 & 0.589 & 0.575 \\
FPS & 0.651 & 0.561 & 0.559 & 0.879 & 0.561 & 0.591 \\
\midrule
GrassTop & \textbf{0.8788} & \textbf{0.8165} & \textbf{0.8232} & \textbf{0.9588} & \textbf{0.8165} & \textbf{0.8629} \\
\bottomrule
\end{tabular}
\end{table}

Within each dataset, NVM is the strongest of the five reference methods on five or six of
the six metrics; the exception is AUC-ROC on NCBI 2022, 2024, and 2024 (All), where
FFP-JS is marginally higher than NVM. Measured against the strongest reference value in
each column, GrassTop's margin depends more on the metric than on the dataset. The margin
is smallest and most stable for AUC-ROC, ranging from $0.0198$ (NCBI 2020) to $0.0238$
(NCBI 2024 All), because all five reference methods already score above $0.87$ on this
metric. The margin is largest for balanced accuracy, F1, and recall, ranging from about
$0.076$ to $0.088$ across the four datasets; balanced accuracy and recall take identical
values in every row because macro-averaged recall coincides with balanced accuracy under
this evaluation protocol. The margin for accuracy and precision falls between the other
two groups, at roughly $0.052$--$0.070$ for accuracy and $0.065$--$0.076$ for precision.

These comparisons concern the complete GrassTop representation, including the selected
descriptors, normalization, retained rank, and chordal distance. Because the published
methods use different feature constructions, the results do not isolate the effect of the
subspace representation alone.

\FloatBarrier
\subsubsection{Phylogenetic clustering purity}
\label{sec:phylogenetics}

Each of the four GrassTop trees has an average purity of 1.0 under the stated purity
measure. Table~\ref{tab:purity} places these results alongside published alignment-free
reference values. For SARS-CoV-2, mammalian mitochondrial genomes, and HRV, the reference
values were read from the bars in Fig.~2d--f of \cite{cakl2025}. The Ebolavirus values are
from Table~2 of \cite{hozumi2024kmertopology} and use a different clade-consistency
measure; they are included for context but not as a direct numerical comparison.

GrassTop is above the five displayed reference values for SARS-CoV-2 and mammalian
mitochondrial genomes and equals the FPS value of 1.0 for HRV. For Ebolavirus, GrassTop
reaches 1.0 under the purity measure used here, while each reference method reaches 1.0
under the related measure reported in \cite{hozumi2024kmertopology}. The equality of these
Ebolavirus values does not establish equivalent performance because the measures differ.

\begin{table}[ht]
\centering
\caption{UPGMA clustering scores. Unmarked values use the tree purity in
Eq.~\eqref{eq:tree-purity}; baseline values for the first three datasets are read from
Fig.~2d--f of \cite{cakl2025}. Values marked $\dagger$ are the clade-consistency precision
scores reported in Table~2 of \cite{hozumi2024kmertopology} and follow
Eq.~\eqref{eq:clade-precision}. GrassTop uses the Betti-0 descriptor type and rank $r=10$.
Bold identifies the GrassTop value; daggered and unmarked values should not be compared
as if they used the same score.}
\label{tab:purity}
\setlength{\tabcolsep}{10pt}
\begin{tabular}{llr}
\toprule
Dataset & Method & Reported score \\
\midrule
\multirow{6}{*}{SARS-CoV-2}
 & NVM & 0.372 \\
 & FFP-JS & 0.835 \\
 & FFP-KL & 0.795 \\
 & FPS & 0.267 \\
 & Markov (MKS) & 0.828 \\
 & GrassTop & \textbf{1.0000} \\
\midrule
\multirow{6}{*}{Mammalian mitochondrial}
 & NVM & 0.906 \\
 & FFP-JS & 0.876 \\
 & FFP-KL & 0.876 \\
 & FPS & 0.953 \\
 & Markov (MKS) & 0.732 \\
 & GrassTop & \textbf{1.0000} \\
\midrule
\multirow{6}{*}{Human rhinovirus (HRV)}
 & NVM & 0.878 \\
 & FFP-JS & 0.745 \\
 & FFP-KL & 0.795 \\
 & FPS & 1.000 \\
 & Markov (MKS) & 0.337 \\
 & GrassTop & \textbf{1.0000} \\
\midrule
\multirow{6}{*}{Ebolavirus}
 & NVM$^\dagger$ & 1.00 \\
 & FFP-JS$^\dagger$ & 1.00 \\
 & FFP-KL$^\dagger$ & 1.00 \\
 & FPS$^\dagger$ & 1.00 \\
 & Markov$^\dagger$ & 1.00 \\
 & GrassTop & \textbf{1.0000} \\
\bottomrule
\end{tabular}
\end{table}

\FloatBarrier

\begin{figure}[p]
\centering
\includegraphics[width=0.47\textwidth]{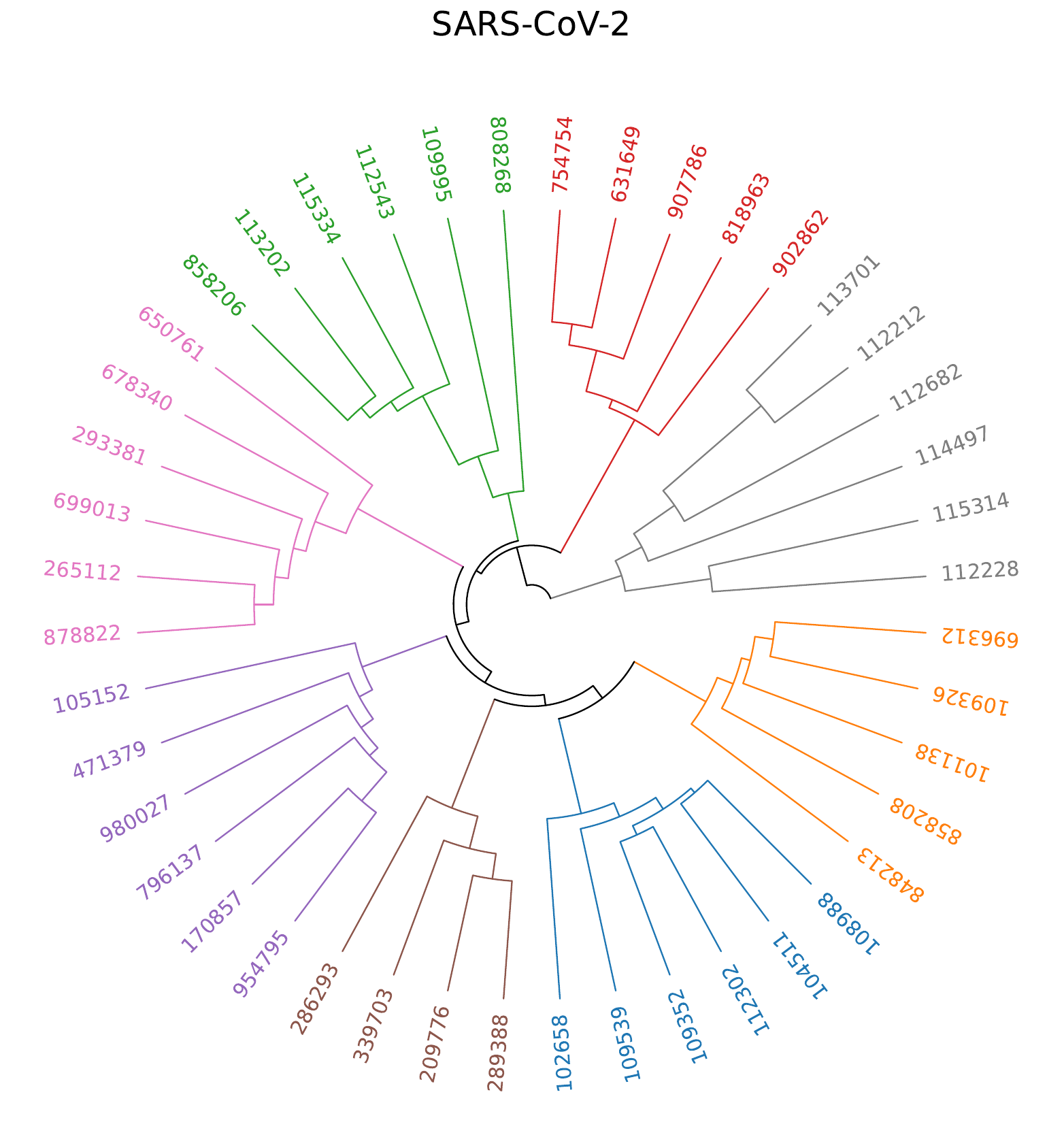}
\hfill
\includegraphics[width=0.47\textwidth]{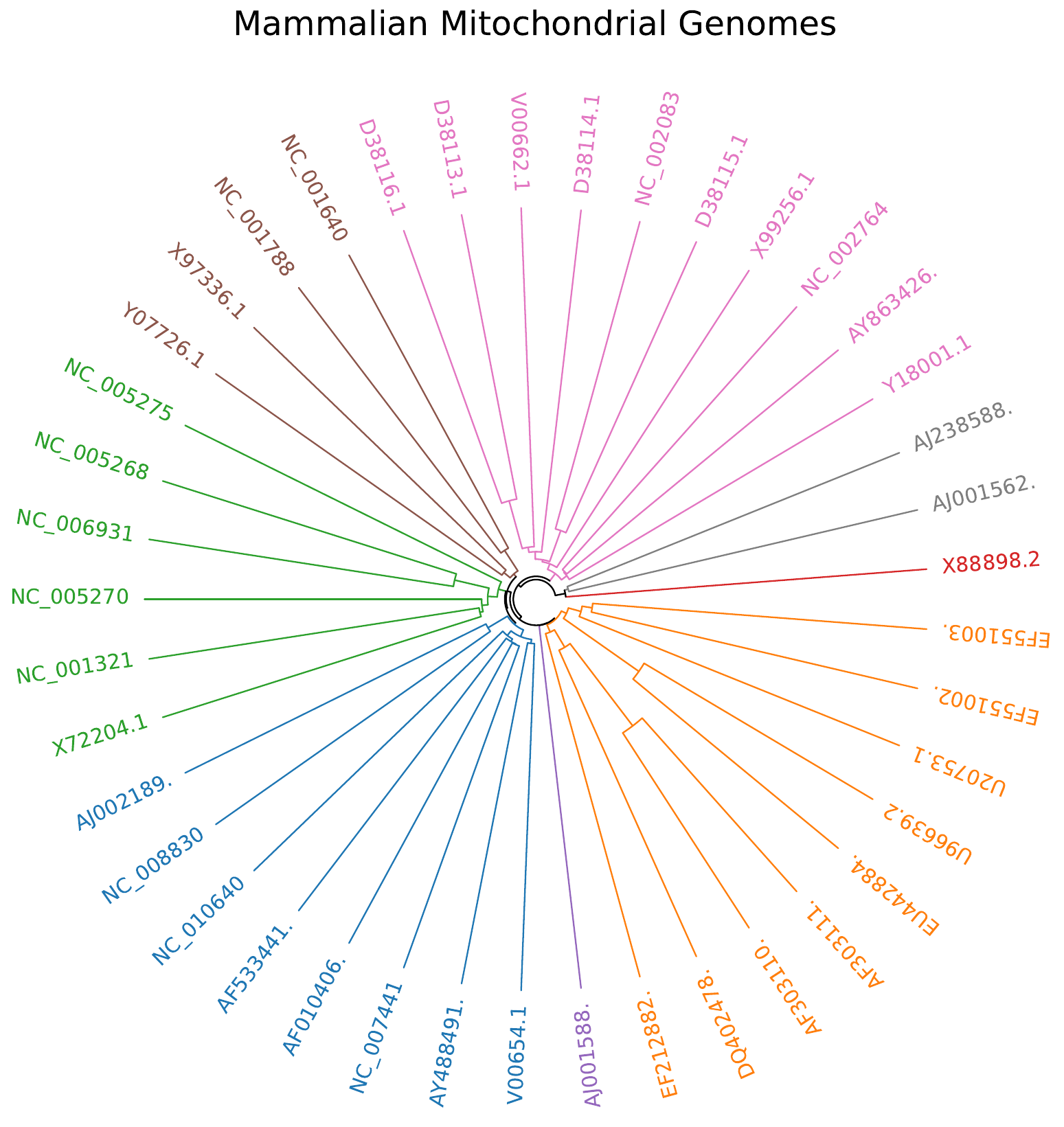}

\vspace{1em}
\includegraphics[width=0.47\textwidth]{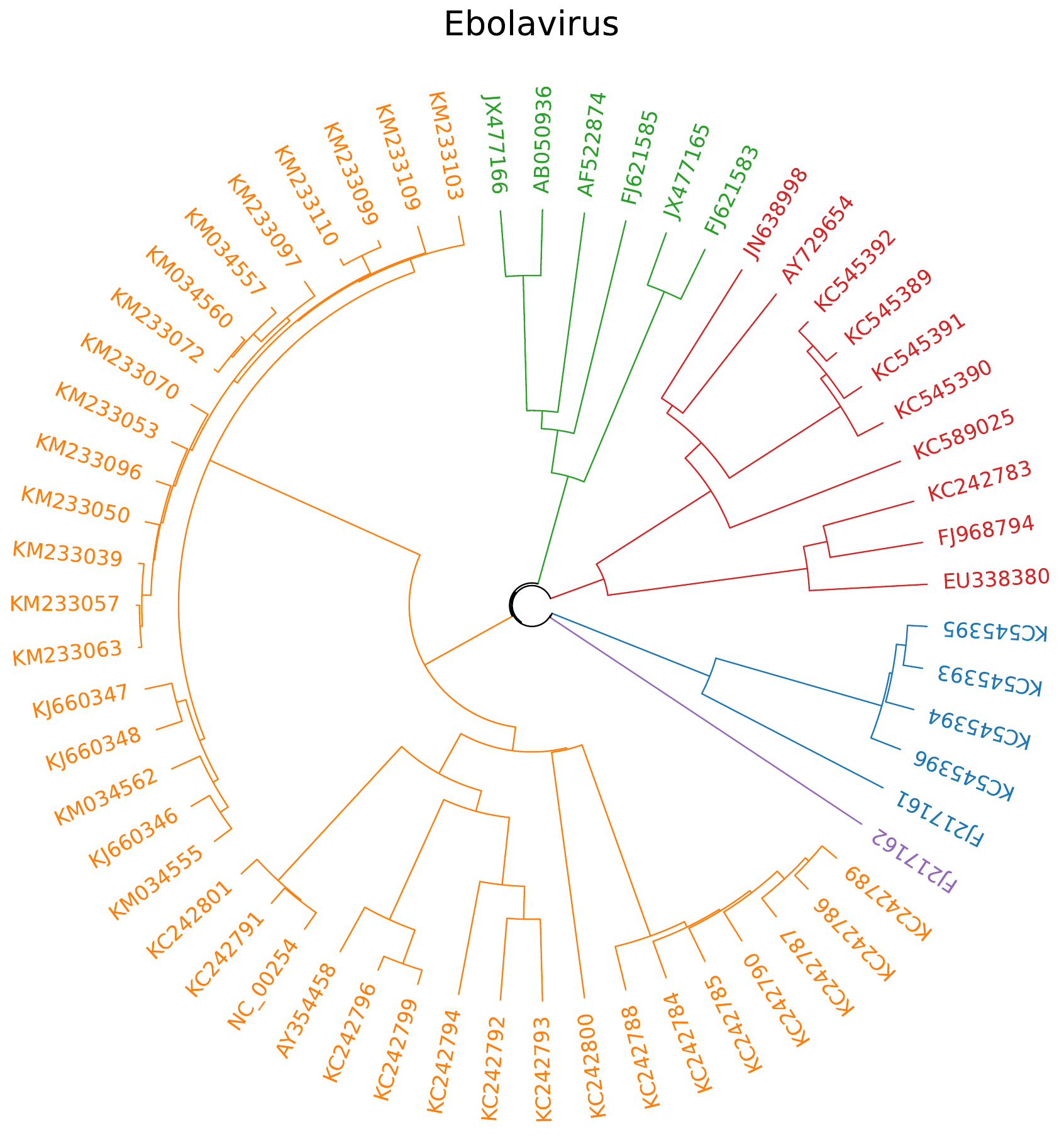}
\hfill
\includegraphics[width=0.47\textwidth]{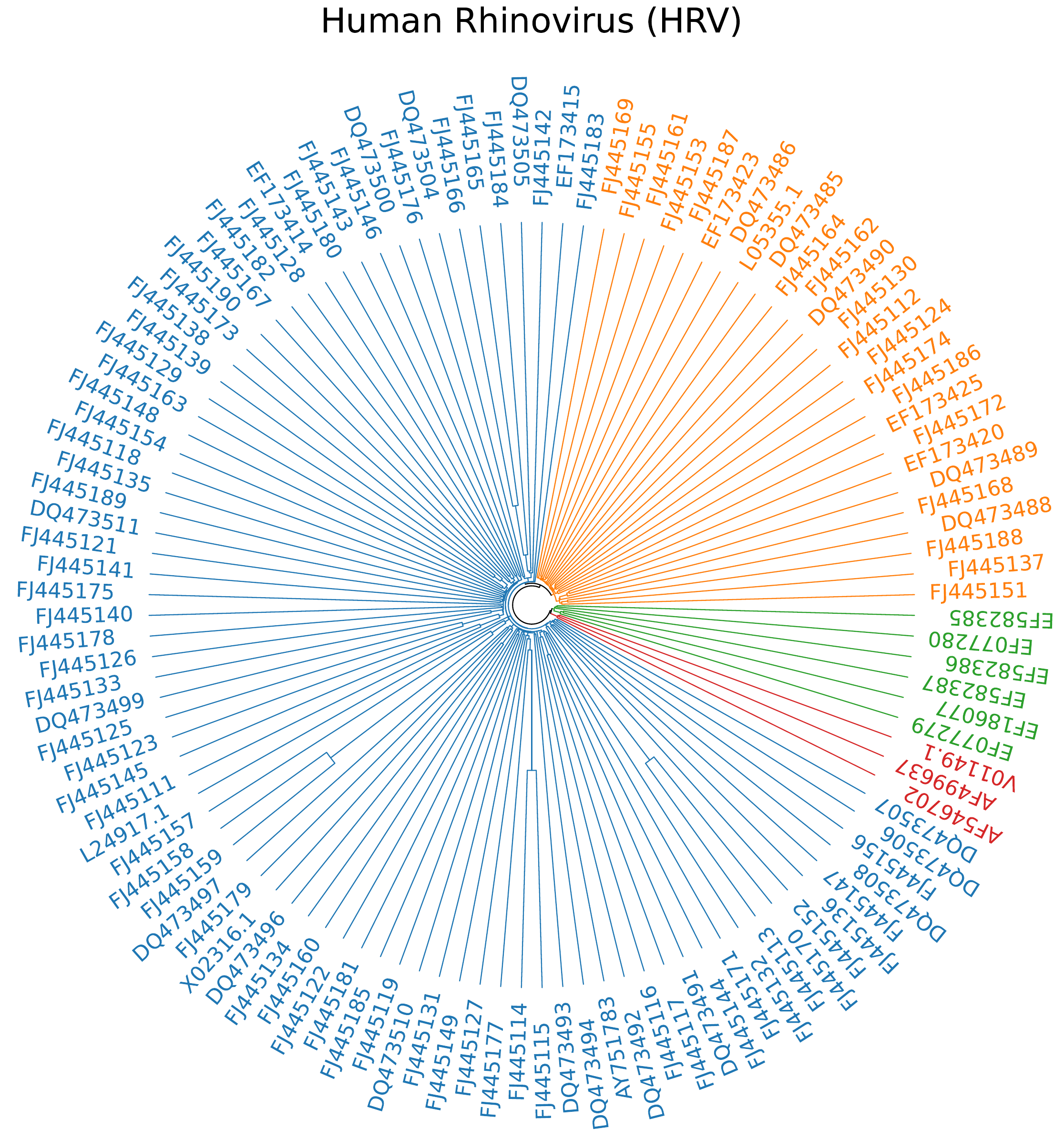}
\caption{Circular UPGMA trees constructed from GrassTop distances. Branches are colored
by the ground-truth label of the largest pure subtree to which they belong and are black
when a merge joins different labels. Top-left: SARS-CoV-2 ($n{=}44$, lineage). Top-right:
mammalian mitochondrial genomes ($n{=}41$, order). Bottom-left: Ebolavirus ($n{=}59$,
species: BDBV/EBOV/RESTV/SUDV/TAFV). Bottom-right: human rhinovirus ($n{=}116$,
HRV-A/HRV-B/HRV-C and the HEV outgroup).}
\label{fig:trees-circular}
\end{figure}

\FloatBarrier
\subsubsection{Stability under sequence perturbation}
\label{sec:stability}

Across the four selected perturbation rates in Appendix
Tables~\ref{tab:label-stability} and~\ref{tab:self-stability} (32 dataset-rate-measure
combinations in total: 4 datasets $\times$ 4 rates $\times$ 2 measures), GrassTop matches
or exceeds comparison based on the Frobenius distance between the unprojected feature
matrices in 29 of 32 combinations, with a strict advantage in
19 and a tie in 10. The three exceptions are small in absolute size and occur either at
the highest tested perturbation rate, where both representations are already heavily
degraded (SARS-CoV-2 label-preservation at rate $0.40$, a $0.02$ gap, and self-preservation
at rate $0.40$, a $0.09$ gap), or in the 20-sequence HRV sample (label-preservation at rate
$0.25$). Figure~\ref{fig:stability-curves} shows label- and self-preservation across the
full 11 perturbation rates.

\begin{figure}[ht]
\centering
\includegraphics[width=0.98\textwidth]{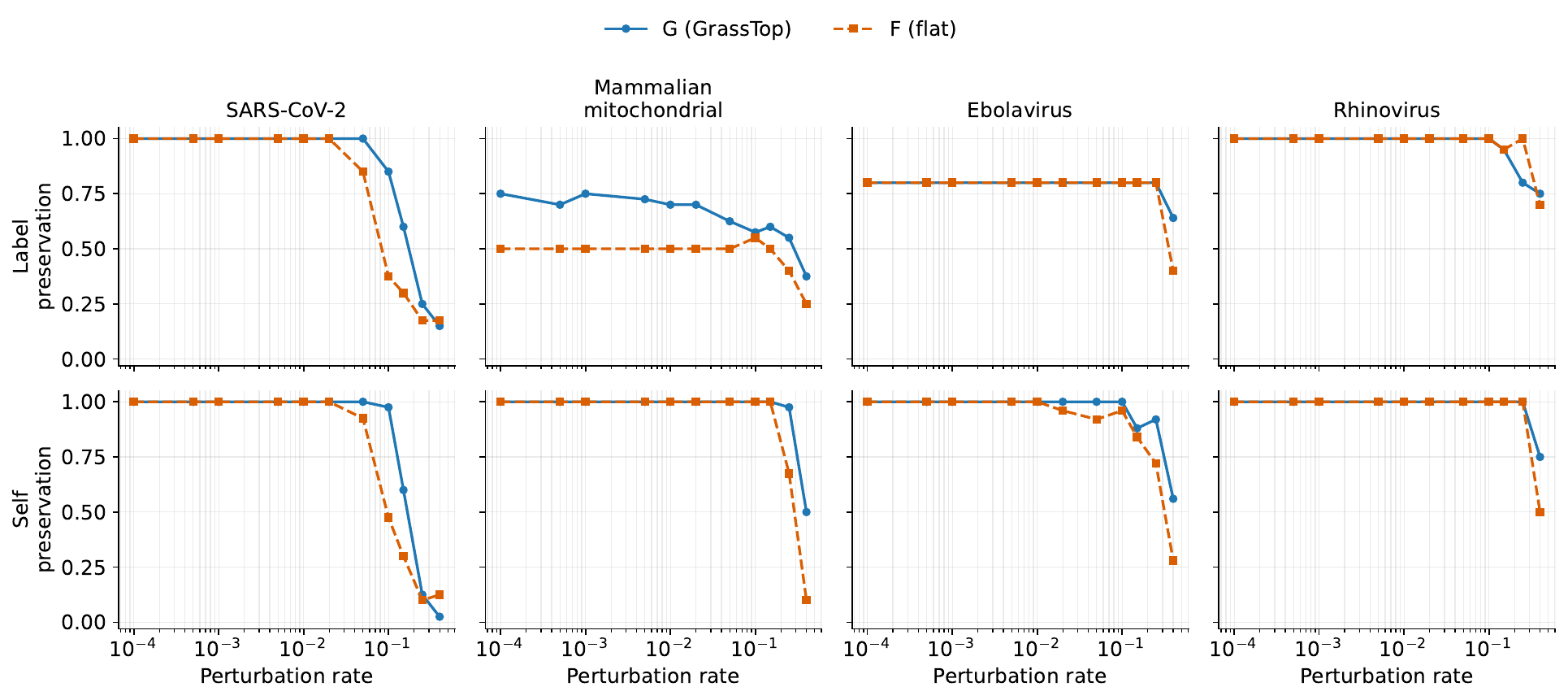}
\caption{Label- and self-preservation over the 11 perturbation rates. Label-preservation
(top) is the fraction of perturbed sequences whose 5-NN majority vote recovers the original label;
self-preservation (bottom) is the fraction for which the original sequence remains among
the five nearest neighbors. The perturbation-rate axis is logarithmic. Solid curves denote
GrassTop and dashed curves denote comparison based on the Frobenius distance between
the unprojected feature matrices.}
\label{fig:stability-curves}
\end{figure}

For SARS-CoV-2, GrassTop has higher preservation over most of the perturbation-rate grid. At
rate $0.10$, label-preservation is $0.85$ for GrassTop and $0.375$ for comparison based
on the unprojected feature matrices, while self-preservation is $0.975$ and $0.475$,
respectively. At rate $0.40$, label-preservation is low for both comparisons and is
slightly higher when the unprojected feature matrices are compared by Frobenius distance
($0.17$ versus $0.15$). For mammalian mitochondrial genomes,
GrassTop has higher label-preservation at every sampled rate; the values at rate $0.25$
are $0.55$ and $0.40$. Values below 1 at the lowest perturbation rates can occur because
preservation is evaluated by 5-NN within a pool of 41 sequences assigned to eight labels;
they do not necessarily result from the perturbation.

For Ebolavirus, both representations have label-preservation of $0.80$ from rate
$0.0001$ through $0.25$. GrassTop has higher self-preservation from rate $0.02$ onward
and is higher on both measures at rate $0.40$. For HRV, the preservation values are equal
through rate $0.15$. At rate $0.25$, label-preservation is $1.00$ for comparison based
on the Frobenius distance between the unprojected feature matrices and $0.80$ for
GrassTop; at rate $0.40$, GrassTop is higher on both
measures. Each HRV value is based on 20 perturbed sequences, so these differences do not locate a
precise crossover rate.

The normalized distances to the original sequences provide a complementary comparison
(Figure~\ref{fig:self-distance}; Appendix Table~\ref{tab:self-distance-values}). In all
four datasets, the normalized chordal distance (GrassTop) falls below the normalized
Frobenius distance between the unprojected feature matrices once the perturbation rate is
high enough, consistent with the upper bound $\sqrt{r}$ for the chordal distance and the
absence of a corresponding upper bound for the Frobenius distance. The rate at which this
crossover occurs, and the size of the
gap it opens by rate $0.40$, both depend on the dataset. For SARS-CoV-2 the crossover is
earliest, at rate $0.0050$, and the gap is largest by rate $0.40$: a ratio of about $3.1$
($13.315$ versus $4.262$). For Ebolavirus the crossover occurs at rate $0.15$, and for
mammalian mitochondrial genomes and HRV it occurs later, at rate $0.25$; the ratio at rate
$0.40$ is correspondingly smaller: about $1.5$ for Ebolavirus ($1.973$ versus $1.341$),
$1.3$ for mammalian mitochondrial genomes ($1.325$ versus $1.057$), and $1.2$ for HRV
($1.261$ versus $1.051$). Because the datasets differ in sequence length, label structure,
and within-group diversity, this experiment does not determine why the crossover rate and
the resulting gap size differ across datasets.

\begin{figure}[ht]
\centering
\includegraphics[width=0.98\textwidth]{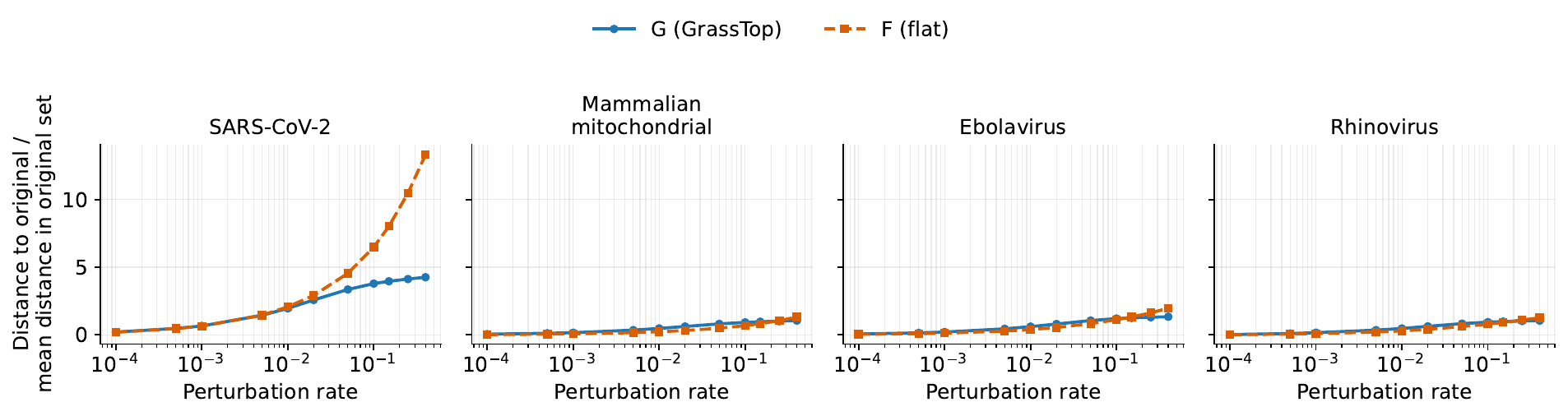}
\caption{Mean distance from each perturbed sequence to its original sequence, normalized by the mean
pairwise distance among the original sequences. The perturbation-rate axis is
logarithmic. Solid curves denote GrassTop and dashed curves denote comparison based on
the Frobenius distance between the unprojected feature matrices;
shaded regions show the interquartile range at each rate.}
\label{fig:self-distance}
\end{figure}
\FloatBarrier

\section{Discussion}
\label{sec:discussion}

Across the three empirical evaluations in Section~\ref{sec:experiments}, GrassTop matches
or exceeds the compared alignment-free methods on nearly every reported measure: it is
above all five reference methods on each of the six classification metrics on all four
NCBI corpora (Section~\ref{sec:classification}); it attains average label purity 1.0 on
all four phylogenetic datasets (Section~\ref{sec:phylogenetics}); and under sequence
perturbation, it matches or exceeds comparison based on the Frobenius distance between
the unprojected feature matrices in 29 of the 32 selected
dataset-rate-measure combinations reported in Section~\ref{sec:stability}. These results
evaluate the complete GrassTop representation as a unit, with the classification and
phylogenetic baselines coming from separately published implementations rather than a
shared codebase. A natural next step is an ablation that holds the descriptors,
normalization, and rank fixed while varying only the subspace/distance step, which would
isolate how much of the observed advantage is attributable to the Grassmannian embedding
itself.

Two directions can further strengthen this evidence. First, the tree purity of
Eq.~\eqref{eq:tree-purity} is a targeted measure of label recovery; extending the
evaluation with branch-length- or divergence-time-aware criteria against established
evolutionary models is a natural next step, complementing the label-recovery evidence
reported here. Second, the distance behavior in the perturbation experiment is consistent
with the different bounds of the two distance functions. The chordal distance is bounded
above by $\sqrt{r}$ (Section~\ref{sec:methods-grassmann}), whereas the Frobenius distance
between the unprojected feature matrices has no corresponding upper bound. In all four
datasets, GrassTop's normalized distance to the original sequence eventually grows more
slowly than that obtained by direct comparison of the unprojected feature matrices. The
perturbation rate at which this crossover occurs and the size of the resulting gap differ
sharply by dataset; identifying what drives this dataset dependence is a further
direction worth pursuing.

The scope of the evidence is defined by the reported configuration and protocols. The experiments use $r=10$ and task-specific descriptor-type sets; other choices are allowed by the
framework but are not evaluated here. The classification baselines are published values
rather than results from a shared implementation, and several phylogenetic values were
estimated from a published figure. The normalization scales are computed from each full
classification corpus without using labels, so held-out observations contribute to this
preprocessing step. Finally, preservation proportions are based on 20--40 perturbed
records per rate and dataset. These conditions limit the comparison to the datasets and
settings reported in this study.

\section{Conclusion}
\label{sec:conclusion}
We presented GrassTop as a 
multiscale topological and spectral method for viral classification and phylogenetic analysis. GrassTop captures $k$-mer positional patterns of a sequence  by a subspace on a Grassmannian. This formulation permits genome comparison with
distances on the Grassmannian and is not restricted to the descriptor types, $k$-mer orders, rank, or distance used in the reported experiments. In these experiments, GrassTop exceeds five published reference methods to cross the viral-family 
classification results displayed and recovers the groups supplied with average purity 1.0 in four UPGMA datasets. Under sequence perturbation, it matches or exceeds comparison based on the Frobenius distance between the unprojected feature matrices in the large majority of the tested dataset-rate-measure combinations (29 of 32), with the size of the advantage varying by dataset. Taken together, the results support GrassTop for 
viral classification and phylogenetic analysis tasks examined here while keeping its conclusions within the tested configurations and datasets. GrassTop can be applied to a variety of tasks involving genome representation and analysis, including antigenic analysis of seasonal virus vaccines and protein-nucleic acid binding prediction.



\section*{Data and Code Availability}

The code used in this study is available at
\url{https://github.com/XiangXiangJY/GrassTop}. The public benchmark datasets can
be obtained from the CAKL repository at
\url{https://github.com/FaisalSuwayyid/CAKL}. SARS-CoV-2 sequences are subject
to the GISAID Terms of Use and are not redistributed.

\section*{Acknowledgments}
This work was supported in part by NIH grant R01AI164266, The University of Georgia (UGA), and Georgia Research Alliance.

\appendix
\renewcommand{\thetable}{A\arabic{table}}
\renewcommand{\thefigure}{A\arabic{figure}}
\renewcommand{\theHtable}{A\arabic{table}}
\renewcommand{\theHfigure}{A\arabic{figure}}
\setcounter{table}{0}
\setcounter{figure}{0}

\section{Numerical results for the sequence perturbation experiment}
\label{app:stability-details}

This appendix gives numerical values that complement the perturbation curves in the main
text. Tables~\ref{tab:label-stability} and~\ref{tab:self-stability} contain label- and
self-preservation at four selected perturbation rates from Figure~\ref{fig:stability-curves}.
Table~\ref{tab:self-distance-values} contains the normalized mean distances plotted in
Figure~\ref{fig:self-distance} for all 11 rates. The evaluation definitions and sampling
design are given in Section~\ref{sec:eval-stability}.

\begin{table}[ht]
\centering
\small
\caption{Label-preservation at four selected perturbation rates. Values are the fractions of
perturbed sequences for which the 5-NN majority vote recovers the original label.
G denotes GrassTop, and F denotes comparison based on the Frobenius distance between the
unprojected feature matrices.}
\label{tab:label-stability}
\setlength{\tabcolsep}{4pt}
\begin{tabular}{lrrrrrrrr}
\toprule
& \multicolumn{2}{c}{rate$=0.05$} & \multicolumn{2}{c}{rate$=0.15$}
& \multicolumn{2}{c}{rate$=0.25$} & \multicolumn{2}{c}{rate$=0.4$} \\
Dataset & G & F & G & F & G & F & G & F \\
\midrule
SARS-CoV-2       & 1.00 & 0.85 & 0.60 & 0.30 & 0.25 & 0.17 & 0.15 & 0.17 \\
Mito.\ (mammal.) & 0.62 & 0.50 & 0.60 & 0.50 & 0.55 & 0.40 & 0.38 & 0.25 \\
Ebolavirus       & 0.80 & 0.80 & 0.80 & 0.80 & 0.80 & 0.80 & 0.64 & 0.40 \\
Rhinovirus       & 1.00 & 1.00 & 0.95 & 0.95 & 0.80 & 1.00 & 0.75 & 0.70 \\
\bottomrule
\end{tabular}
\end{table}

\begin{table}[ht]
\centering
\small
\caption{Self-preservation at four selected perturbation rates. Values are the fractions of
perturbed sequences whose original sequence remains among their five nearest neighbors. G denotes
GrassTop, and F denotes comparison based on the Frobenius distance between the
unprojected feature matrices.}
\label{tab:self-stability}
\setlength{\tabcolsep}{4pt}
\begin{tabular}{lrrrrrrrr}
\toprule
& \multicolumn{2}{c}{rate$=0.05$} & \multicolumn{2}{c}{rate$=0.15$}
& \multicolumn{2}{c}{rate$=0.25$} & \multicolumn{2}{c}{rate$=0.4$} \\
Dataset & G & F & G & F & G & F & G & F \\
\midrule
SARS-CoV-2       & 1.00 & 0.93 & 0.60 & 0.30 & 0.12 & 0.10 & 0.03 & 0.12 \\
Mito.\ (mammal.) & 1.00 & 1.00 & 1.00 & 1.00 & 0.97 & 0.68 & 0.50 & 0.10 \\
Ebolavirus       & 1.00 & 0.92 & 0.88 & 0.84 & 0.92 & 0.72 & 0.56 & 0.28 \\
Rhinovirus       & 1.00 & 1.00 & 1.00 & 1.00 & 1.00 & 1.00 & 0.75 & 0.50 \\
\bottomrule
\end{tabular}
\end{table}
\FloatBarrier

\begin{table}[ht]
\centering
\small
\caption{Normalized mean distance from each perturbed sequence to its original sequence over all 11
perturbation rates in Figure~\ref{fig:self-distance}. Each value is divided by the mean
pairwise distance among the original sequences and rounded to three decimal
places. G denotes GrassTop, and F denotes comparison based on the Frobenius distance
between the unprojected feature matrices. The numbers of perturbed sequences per
rate are 40, 40, 25, and 20 for SARS-CoV-2, mammalian mitochondrial genomes, Ebolavirus,
and HRV, respectively.}
\label{tab:self-distance-values}
\setlength{\tabcolsep}{3.5pt}
\begin{tabular}{lrrrrrrrr}
\toprule
& \multicolumn{2}{c}{SARS-CoV-2} & \multicolumn{2}{c}{Mito.} &
\multicolumn{2}{c}{Ebolavirus} & \multicolumn{2}{c}{Rhinovirus} \\
Rate & G & F & G & F & G & F & G & F \\
\midrule
0.0001 & 0.206 & 0.198 & 0.052 & 0.023 & 0.065 & 0.040 & 0.023 & 0.014 \\
0.0005 & 0.477 & 0.466 & 0.119 & 0.053 & 0.149 & 0.091 & 0.100 & 0.061 \\
0.0010 & 0.658 & 0.645 & 0.161 & 0.074 & 0.208 & 0.129 & 0.163 & 0.096 \\
0.0050 & 1.445 & 1.458 & 0.353 & 0.163 & 0.448 & 0.280 & 0.352 & 0.211 \\
0.0100 & 1.958 & 2.071 & 0.477 & 0.230 & 0.610 & 0.393 & 0.474 & 0.294 \\
0.0200 & 2.580 & 2.915 & 0.621 & 0.321 & 0.801 & 0.545 & 0.632 & 0.412 \\
0.0500 & 3.356 & 4.552 & 0.814 & 0.492 & 1.061 & 0.832 & 0.832 & 0.612 \\
0.1000 & 3.791 & 6.485 & 0.925 & 0.683 & 1.203 & 1.119 & 0.940 & 0.795 \\
0.1500 & 3.963 & 8.026 & 0.970 & 0.827 & 1.257 & 1.332 & 0.985 & 0.918 \\
0.2500 & 4.127 & 10.479 & 1.015 & 1.051 & 1.302 & 1.623 & 1.026 & 1.102 \\
0.4000 & 4.262 & 13.315 & 1.057 & 1.325 & 1.341 & 1.973 & 1.051 & 1.261 \\
\bottomrule
\end{tabular}
\end{table}
\FloatBarrier

\bibliographystyle{plain}
\bibliography{refs}

\end{document}